\documentclass[11pt]{article}

\usepackage{cite}

\usepackage{cite}

\usepackage{varwidth}

\usepackage[normalem]{ulem}
\usepackage{tensor}
\usepackage{emptypage}
\usepackage{slashed}
\usepackage{color}
\usepackage{graphicx}  
\usepackage{psfrag}
\usepackage{tocbibind} 
\usepackage{epsf}
\usepackage{epic}
\usepackage{eepic}
\usepackage{epsfig}
\usepackage{wrapfig}
\usepackage{a4}
\usepackage{amssymb,latexsym}
\usepackage[mathscr]{eucal}
\usepackage[active]{srcltx}
\usepackage{url}
\usepackage{fancybox}
\usepackage{ntheorem}
\usepackage{mhequ}
\usepackage[g]{esvect}
\usepackage{epstopdf}

\usepackage{wasysym}

\usepackage{amsfonts}
\usepackage{multirow}

\usepackage{tikz}  \usetikzlibrary{patterns}

\usepackage{bbm}
\usepackage{slantsc}
\usepackage{slashed}

\usepackage[normalem]{ulem}

\usepackage{marginnote}

\newcommand{\blue}[1]{{\color{blue}#1}}
\newcommand{\red}[1]{{{\color{red}#1}}}

\newcommand{\secEdbh}[1]{\mnote{ptcSecEdbh environement, for a second edition in the black hole book, but this means that the stuff around has been used in the black hole notes, so the whole file should be discarded or replaced by a reference?}}

\newcommand{\ptcSecEd}[1]{\mnote{ptcSecEd environement, but this means that this has been used in the Vienna lecture notes, so the whole file should be discarded or replaced by a reference?}}

\newcommand{\tooslow}[1]{\ptc{tooslow environement; figures commented out because too long to display and compile, trying to reduce them with acrobat did not give anything useful; should be fixed in the cleanrun file}}

\newcommand{\rnoblackholesc}[2]{}
\newcommand{\waehr}[1]{}
\newcommand{\licht}[1]{}

\newcommand{\docendnp}{\ptcr{remove end document}\newpage\input{end}\end{document}}
\newcommand{\docend}{\ptcr{remove end document}\newpage\input{endWithProblems}\end{document}}

\newcommand{\ptcn}[1]{\ptcr{#1}}

\newcommand{\syncx}[1]{\ptcxx{here starts a syncx, or possibly a restrict environment}}
\renewcommand{\syncx}[1]{\ptc{syncx command here}{\color{red}#1}}

\newcommand{\roscoff}[1]{\ptc{roscoff here}}

\newcommand{\ptcxx}[1]{\mnote{{\bf ptcss:} {\color{red} #1}}}

\newcommand{\mnotex}[1]
{\protect{\stepcounter{mnotecount}}$^{\mbox{\footnotesize
$
\bullet$\themnotecount}}$ \marginpar{
\raggedright\tiny\em
$\!\!\!\!\!\!\,\bullet$\themnotecount: #1} }

\newcommand{\jamesx}[1]{}
\renewcommand{\jamesx}[1]{{\mnote{{\color{blue}{\bf jg:}
#1} }}}

\newcommand{\tnabla}{\widetilde \nabla}

\newcommand{\tcM}{\bmcM}

\newcommand{\bmcM}{\red{\mbox{\dangerous} \,\,\,\,\widetilde{\!\!\!\!\mcM}}}

\newcommand{\doc}{\dangerous\langle\langle \scri \rangle\rangle}

\newcommand{\fourgKK}{\red{g}}
\newcommand{\dangerous}{{\fontencoding{U}\fontfamily{futs}\selectfont\char 66\relax}}

\def\ben{\begin{equation}}
\def\een{\end{equation}}
\def\bena{\begin{eqnarray}}
\def\eena{\end{eqnarray}}

\def\f(#1/#2){\frac{#1}{#2}}
\def\Frac(#1/#2){\left(\frac{#1}{#2}\right)}
\def\chris(#1-#2-#3){{\mit \Gamma}^{#1}{}_{{#2}{#3}} }
\def\tilchris(#1-#2-#3){\tilde{{\mit \Gamma}}^{#1}{}_{{#2}{#3}}}
\def\hatchris(#1-#2-#3){\hat{{\mit \Gamma}}^{#1}{}_{{#2}{#3}}}

{\catcode `\@=11 \global\let\AddToReset=\@addtoreset}
\AddToReset{figure}{section}

\DeclareFontFamily{OT1}{rsfs}{}
\DeclareFontShape{OT1}{rsfs}{m}{n}{ <-7> rsfs5 <7-10> rsfs7 <10-> rsfs10}{}
\DeclareMathAlphabet{\mycal}{OT1}{rsfs}{m}{n}

{\catcode `\@=11 \global\let\AddToReset=\@addtoreset}
\AddToReset{equation}{section}

\newcounter{mnotecount}[section]

\renewcommand{\themnotecount}{\thesection.\arabic{mnotecount}}

\newcommand{\mnote}[1]
{\protect{\stepcounter{mnotecount}}$^{\mbox{\footnotesize
$
\bullet$\themnotecount}}$ \marginpar{
\raggedright\tiny\em
$\!\!\!\!\!\!\,\bullet$\themnotecount: #1} }

\newcommand{\ptcr}[1]{{\color{red}\mnote{{\color{red}{\bf ptc:}
#1} }}}
\newcommand{\jj}[1]%
{{\color{red}\mnote{{\color{red}{\bf jj:} #1} }}}

\definecolor{HP}{rgb}{1,0.09,0.58}

\newcommand{\eqref}[1]{\eq{#1}}

\newcommand{\hs}{\cH_{\mbox{\scriptsize sing}}}

\newcommand{\bmcD}{\,\,\widetilde{\!\!\mcD}}%

\newcommand{\beadl}[1]{\begin{deqarr}\label{#1}}
\newcommand{\eeadl}[1]{\arrlabel{#1}\end{deqarr}}%

\def\nz{\ifmmode {I\hskip -3pt N} \else {\hbox {$I\hskip -3pt N$}}\fi}
\def\zz{\ifmmode {Z\hskip -4.8pt Z} \else
       {\hbox {$Z\hskip -4.8pt Z$}}\fi}
\def\qz{\ifmmode {Q\hskip -5.0pt\vrule height6.0pt depth 0pt
       \hskip 6pt} \else {\hbox
       {$Q\hskip -5.0pt\vrule height6.0pt depth 0pt\hskip 6pt$}}\fi}
\def\rz{\ifmmode {I\hskip -3pt R} \else {\hbox {$I\hskip -3pt R$}}\fi}
\def\cz{\ifmmode {C\hskip -4.8pt\vrule height5.8pt\hskip 6.3pt} \else
       {\hbox {$C\hskip -4.8pt\vrule height5.8pt\hskip 6.3pt$}}\fi}
\def\au{{\setbox0=\hbox{\lower1.36775ex\hbox{''}\kern-.05em}\dp0=.36775ex\hs
kip0pt\box0}}
\def\ao{{}\kern-.10em\hbox{``}}

\newcommand\Gregbeq{\begin{eqnarray}}
\newcommand\Gregeeq{\end{eqnarray}}

\newcommand{\scri}{{\mycal I}}%
\newcommand{\scrip}{\scri^{+}}%
\newcommand{\scrim}{\scri^{-}}%
\newcommand{\Scri}{\scri}

\def\cH{{\cal H}}

\def\h1{{\hat 1}}
\def\h2{{\hat 2}}

\def\3f{\frac{3}{2}}

\newcommand{\oversetty}[2]{%
\mathop{#2}\limits^{\vbox to -.1ex{%
\kern -1.5ex\hbox{$\scriptstyle #1$}\vss}}}

\newcommand{\jlcax}[1]{}
\newcommand{\eean}{\nonumber\end{eqnarray}}

\newcommand{\kk}[1]{}

\newcommand{\beq}{\begin{equation}}

\newcommand{\rgc}[1]{}

\newcommand{\FS}       
                  {F}

\newcommand{\HS} 
       {H_{\mbox{\scriptsize volume}}}

{\ptc{this should be removed in the oberwolfach version}}%

\newcommand{\eel}[1]{\label{#1}\end{equation}}
\newcommand{\eeal}[1]{\label{#1}\end{eqnarray}}
\newcommand{\bed}{\begin{deqarr}}
\newcommand{\eed}{\end{deqarr}}
\newcommand{\bedl}[1]{\begin{deqarr}\label{#1}}
\newcommand{\eedl}[2]{\arrlabel{#1}\label{#2}\end{deqarr}}

\newcommand{\tg}{{\widetilde{g}}}

\newcommand{\mcU}{{\mycal U}}

\newcommand{\mcN}{{\mycal N}}
\newcommand{\mcK}{{\mycal K}}

\newcommand{\bel}[1]{\begin{equation}\label{#1}}
\newcommand{\bea}{\begin{eqnarray}}
\newcommand{\bean}{\begin{eqnarray}\nonumber}
\newcommand{\beal}[1]{\begin{eqnarray}\label{#1}}
\newcommand{\eea}{\end{eqnarray}}

\newcommand{\Eqref}[1]{Equation~\eq{#1}}

\newcommand{\qed}{\hfill $\Box$}
\newcommand{\qedskip}{\hfill $\Box$\medskip}
\newcommand{\proof}{\noindent {\sc Proof:\ }}
\newcommand{\be}{\begin{equation}}
\newcommand{\eeq}{\end{equation}}
\newcommand{\ee}{\end{equation}}
\newcommand{\beqa}{\begin{eqnarray}}
\newcommand{\eeqa}{\end{eqnarray}}
\newcommand{\beqan}{\begin{eqnarray*}}
\newcommand{\eeqan}{\end{eqnarray*}}
\newcommand{\ba}{\begin{array}}
\newcommand{\ea}{\end{array}}

\newcommand{\hyp}{\mycal S}

\newcommand{\mcM}{{\mycal M}}
\newcommand{\mcD}{{\mycal D}}

\newcommand{\warn}[1]
{\protect{\stepcounter{mnotecount}}$^{\mbox{\footnotesize
$
\bullet$\themnotecount}}$ \marginpar{
\raggedright\tiny\em
$\!\!\!\!\!\!\,\bullet$\themnotecount: {\bf Warning:} #1} }

\newcommand{\R}{\mathbb{R}}

\newcommand{\eq}[1]{(\ref{#1})}

\newcommand{\ptc}[1]{\mnote{{\bf ptc:}#1}}

\newcommand{\beqar}{\begin{deqarr}}
\newcommand{\eeqar}{\end{deqarr}}

\newcommand{\beaa}{\begin{eqnarray*}}
\newcommand{\eeaa}{\end{eqnarray*}}

 \let\g=\gamma

\newcommand{\bethm}{\begin{theorem}}
\newcommand{\et}{\end{theorem}}
\newcommand{\bl}{\begin{Lemma}}

\newtheorem{Theorem} {\sc  Theorem\rm} [section]
\newtheorem{theorem} [Theorem] {\sc  Theorem\rm}

\newtheorem{Lemma} [Theorem] {\sc  Lemma\rm}
\newtheorem{Proposition} [Theorem] {\sc  Proposition\rm}

\newtheorem{Definition}[Theorem]{\sc  Definition\rm}

\theoremstyle{Remark}

\theorembodyfont{\upshape}

\newtheorem{Remark}[Theorem]{\sc Remark\rm}

\theoremsymbol{\ensuremath{\diamondsuit}}
\newcommand{\fcoco}{\small}
\theorembodyfont{\fcoco}\theoremseparator{}\theoremindent0.5cm

\theoremstyle{nonumberplain}\theorembodyfont{\fcoco}
\theoremseparator{}\theoremindent0.5cm

\theoremstyle{definition}

\DeclareFontFamily{OT1}{rsfs}{}
\DeclareFontShape{OT1}{rsfs}{m}{n}{ <-7> rsfs5 <7-10> rsfs7 <10-> rsfs10}{}
\DeclareMathAlphabet{\mycal}{OT1}{rsfs}{m}{n}

{\catcode `\@=11 \global\let\AddToReset=\@addtoreset}
\AddToReset{equation}{section}

\renewcommand{\ptcn}[1]{} 
\renewcommand{\rnoblackholesc}[2]{}

\renewcommand{\bmcM}{{\,\,\,\,\widetilde{\!\!\!\!\mcM}}}
\renewcommand{\dangerous}{}
\renewcommand{\red}[1]{{\color{blue} #1}}

\renewcommand{\red}[1]{#1}
\renewcommand{\blue}[1]{#1}

\begin{document}
\title{%
Roads to
topological
censorship\protect\thanks{Preprint UWThPh-2019-17}
}

\author{Piotr T.\ Chru\'{s}ciel
and Gregory J. Galloway}

\maketitle
\begin{abstract}
 We review some aspects of topological censorship.  We present
 {several}  alternative sets of hypotheses which allow the proof of a
topological censorship theorem for spacetimes with conformal completions at infinity and vanishing cosmological constant.
\end{abstract}

\setcounter{tocdepth}{1}
\tableofcontents

\definecolor{HP}{rgb}{1,0.09,0.58}

\section{Introduction}

The beautiful topological censorship argument of Friedman, Schleich and Witt (FSW)~\cite{FriedmanSchleichWitt} has become a standard tool to obtain insight  {into}  the topology of spacetimes. It has been successfully applied to derive topology restrictions in several situations~\cite{Jacobson:venkatarami,Galloway:fitopology,CGS,ChWald,GallowayBrowdy,%
GSWW,CGL,galloway-topology}.

The setting for the  {original} FSW result is that of asymptotically flat spacetimes in the sense of Penrose.   It has been noted in the past (cf.\ e.g.~\cite{galloway:woolgar,GSWW}),
 that some of the arguments in~\cite{FriedmanSchleichWitt}  require somewhat stronger assumptions.
Indeed, the following issues need to be addressed in the proof:

\begin{enumerate}
  \item
  The question of ``$i_0$ avoidance'', namely that the future of a compact subset of the physical spacetime does not contain all of {future null infinity $\scri^+$}. This problem was already discussed in~\cite{galloway:woolgar}  (see in particular the second paragraph on p.~L2 there),  where it was proposed to extend  global hyperbolicity to  {$\scri^+$}). \red{Lemma~\ref{L29III11.1a} shows how this hypothesis  takes care of the problem.}
 \item The need to rule out “outward directed” null geodesics going from one  {$\scri^+$} to another in the covering space argument, which is addressed in~\cite{FriedmanSchleichWitt} but  further conditions seem to be needed to obtain the desired conclusion.
     \red{
     We implement this by adding the condition \eqref{30III19.7} to Theorem~\ref{T22III19.1}. Alternative possibilities are given in Proposition~\ref{P12IV19.1}.
     }
\end{enumerate}

Part of the  above may be traced back to   some claims in~\cite{HE} which require supplementary assumptions, cf.~\cite[Appendix~B]{ChDGH} for a discussion in a related context.

In a version of the theorem where one replaces pointwise energy conditions by   the Average Null Energy Condition (ANEC), one further needs to require ANEC to hold on half-geodesics starting near  {past null infinity $\scri^-$}, and the argument requires suitably perturbing a cross section of  {$\scri^-$} into the physical spacetime.   This was addressed in~\cite{GSWW}, albeit in the context of a timelike Scri.  But the issue is  {essentially} the same in the current case.

We note that  the topological censorship theorems proved in~\cite{CGS} in a Kaluza-Klein setting apply directly to the usual asymptotically flat spacetimes. In particular Theorem~5.5 there applies by taking the group $G$ to be trivial, cf.\ Theorem~\ref{notrapped0} below. However, the relevant part of the analysis in~\cite{CGS} assumes a uniformity property of Scri which might be viewed as too restrictive.

The aim of this note is to propose
some alternative  precise  sets of assumptions, with a causality-theory-flavor, which allow  one to obtain, still by the general approach of~\cite{FriedmanSchleichWitt},
a   topological censorship theorem in the asymptotically flat setting.

The reader is referred to \cite{EGP,BakerGalloway} for related results in a Cauchy data context.

\section{Reminders}
 \label{s9VIII12.1}

The original Friedman-Schleich-Witt censorship theorem is concerned with spacetimes which are asymptotically flat at null infinity. The precise framework is that of conformal completions, as introduced by
Penrose~\cite{penrose:asymptotic}:

\begin{Definition}
\label{DefConfComp+}
A pair $(\tcM ,\tg)$ will be called a
\emph{conformal completion at infinity} of $(\mcM,g)$ if $\tcM $
is a manifold with boundary such that:
\begin{enumerate}\item $\mcM$ is the interior of $\tcM $,
 \item  on $\bmcM$ there exists a function $\Omega$  with the
property that the metric $\tg$, defined as $\Omega^2g$ on
$\mcM$, extends by continuity to the boundary of $\tcM$, with the
extended metric maintaining its signature on the boundary,
\item
$\Omega$ is (strictly) positive on $\mcM$, differentiable on $\tcM $,
vanishes precisely on  $\Scri:=\partial \bmcM$, with $d\Omega$ \emph{nowhere vanishing} on
$\Scri$.
\end{enumerate}
\end{Definition}

In the case where $\scri$ is a null hypersurface, Geroch and Horowitz~\cite{GerochHorowitz} propose to add  the conditions that 1) one can find a gauge such that the Hessian of $\Omega$ vanishes on $\scri$, and that 2) the integral curves of $\tnabla \Omega$ on $\scri$ are complete in this gauge.
While this is certainly useful for some purposes, these conditions do not seem to play any
obvious role in our treatment of the problem at hand.

As usual, we set
\begin{equation}\label{30III19.5}
  \scrip := \scri \cap J^+(\mcM)
   \,,
   \quad
  \scrim := \scri \cap J^-(\mcM)
   \,.
\end{equation}

The domain of outer communications  $\doc$ is defined as
$$
 \doc:=I^+(\scrim) \cap I^{-}(\scrip)
 \,.
$$


Unless explicitly specified otherwise we use the terminology and notation of \cite{CGS}.

For the reader's convenience  we start by recalling \cite[Theorem~5.5]{CGS} with the Kaluza-Klein group $G$ taken to be trivial, as rewritten in  the context of Penrose's definition of asymptotic flatness. For this some terminology will be needed.

We say that $\scri$ satisfies \emph{asymptotic estimates uniform to order one} if there exists
\begin{enumerate}
  \item
a ``radial function'' $r$ and a constant \red{$R_{0}$} such that the set $\{r\ge R_0\}$ forms a neighborhood of $\scri$, with the level sets of $r$ \emph{timelike} there, and
  \item
 a time function $t$ on $\mcM$ such that for all $\red{r}\ge R_0$   and for  all $\tau \in \R$  the sets
$$
 \hyp_{r,\tau}:=\{r=R, t=\tau\}
$$
are \emph{smooth past and future inner trapped compact
spacelike
submanifolds} of $\mcM$.
\end{enumerate}

 The terminology is related to the formula for the expansion scalar of null hypersurfaces,%
which guarantees
the above
 in spacetimes which are asymptotically flat in a coordinate sense (as in \cite{CGS}), \emph{provided that} the metric and its first derivatives  satisfy  bounds which  are uniform in time in the asymptotic region $
 \{r\ge R_0\}$.

 We have the following rewording of Theorem~5.5 in \cite{CGS}:

\begin{Theorem}
\label{TCdsc} Let $(\mcM,\fourgKK)$ be a spacetime satisfying
the null energy condition, i.e.\
\begin{equation}\label{22III19.1}
  R_{\mu\nu}X^\mu X^\nu \ge 0
  \
  \mbox{for all null vector fields $X$,}
\end{equation}
and admitting a conformal completion with null boundary $\scri$ satisfying asymptotic estimates uniform to order one.

If the domain of outer communications $\doc$ is globally hyperbolic  and  if there exists  $R_1\ge R_0$ such that the set
$\{ r\ge R_1\}$  is simply connected, then
 $\doc $ is simply connected.
\qed
\end{Theorem}

We continue with the
following non-visibility result, closely related to the problem  at hand:

\begin{Theorem}\label{notrapped0} Let $(\mcM,\fourgKK)$   admit a
future conformal completion
\red{as in \eqref{30III19.5},
$$
 \bmcM = \mcM \cup \scri^+
 \,,
$$
and assume that $\scri^+$ is a connected null hypersurface.
}
Suppose  that
%
%
$$
 \bmcD := \mcD \cup \scri^+ \,, \ \mbox{ where $\mcD := I^-(\scri^+
  )\cap \mcM$,}
$$
is globally hyperbolic.
If the  null energy condition  holds on $\mcD$, then there are no compact weakly future trapped spacelike
submanifolds of codimension two
within $\mcD$.
\end{Theorem}

Theorem \ref{notrapped0} is~\cite[Theorem~6.1]{CGS} with the (not-made-clear there) notion of ``regular $\scri$''  being captured by Definition~\ref{DefConfComp+}, together with the observation that the
``$i^0$-avoidance condition'' of \cite[Theorem~6.1]{CGS} already follows%
\footnote{%
Note that for spacelike $\scri$'s the property, needed for topological censorship, that the future of a trapped surface does \emph{not} include the whole of $\scrip$, cannot be inferred from causality conditions, see~\cite{CGL}.
}
from hypothesis 1.\ of~\cite[Theorem~6.1]{CGS}:

\begin{Lemma}
 \label{L29III11.1a}
Under the hypotheses of Theorem~\ref{notrapped0},   for any compact set $K \subset \mcD$, $J^+(K, \bmcD)$ does not
 contain all of~$\scri^+$.
  \end{Lemma}

  \begin{proof}
Let $p\in J^+(K,\bmcD)\cap \scrip$. By global hyperbolicity $J^-(p,,\bmcD) \cap J^+(K,\bmcD) $ is compact. Suppose that $J^+(K, \bmcD)$
 contains all of~$\scri^+$, then the generator of $\scrip$ through $p$ is contained in the compact set $J^-(p,,\bmcD) \cap J^+(K,\bmcD) $,
  which is impossible in a globally hyperbolic, hence strongly causal, spacetime.
 \qedskip
\end{proof}

For our  purposes below we will need a variation of
Theorem~\ref{notrapped0}. Consider a compact codimension-two spacelike submanifold of $\mcM$, say $S$, and let us suppose that $S$ is two-sided, in the sense of existence of a well-defined field of spacelike normals. Let us  choose to denote one of those sides ``inner''. Instead of supposing, as in Theorem~\ref{notrapped0}, that both families of null future directed geodesics are weakly trapped at $S$, we will only assume that the inner family is weakly trapped.
The proof of Theorem~\ref{notrapped0} applies as is to this setting to give:

\begin{Theorem}\label{T12IV19.1} Under the hypotheses of Theorem~\ref{notrapped0}, let $S$ be a future weakly inner trapped submanifold contained in $\mcD$. Then there are no  null geodesics on $\dot J^+(S)$ connecting $\scrip$ with the inner side of $S$.
 \qed
\end{Theorem}

\begin{Remark}
As such, we will only need Theorem~\ref{T12IV19.1} for surfaces which are inner trapped rather than weakly inner-trapped.
 In such a case the proof of \cite[Theorem~6.1]{CGS} can be simplified by   invoking
\cite[Propositions 12.32 and 12.33]{beem:ehrlich:global}, or  \cite[Proposition 48, p.\ 296]{BONeill} (compare~\cite[Section~2]{Galloway:fitopology}), rather than using a weak comparison principle for null hypersurfaces to exclude the borderline cases.
\qed
\end{Remark}

\section{The theorem}

We have the following alternative to Theorem~\ref{TCdsc}:

\begin{Theorem}
 \label{T22III19.1}
Let $(\mcM,g)$ be a globally hyperbolic space-time satisfying the null energy condition.
Assume that $(\mcM,g)$ admits a  completion at infinity such that both $\scrip$ and $\scrim$ are connected null hypersurfaces, and that
\begin{equation}\label{15IV11.1}
 \bmcD := \mcD \cup \scri^+
 \,,
  \
  \mbox{ where $\mcD := I^-(\scri^+)\cap \mcM$}
\end{equation}
is globally hyperbolic. Suppose that there exists a simply connected neighborhood $\mcU$ of $\scri$
and
a  foliation of a (perhaps smaller) neighborhood of $\scrim$ with
spacelike   up-to-boundary and smooth  up-to-boundary
 acausal hypersurfaces $\hyp_\tau \subset \bmcM$, $\tau\in R$,
such that for each $\tau$  the intersections $\red{S_{\tau,\varepsilon}}$ of $\hyp_\tau$
with the $\varepsilon$-level sets of $\Omega$ are
\begin{equation}\label{30III19.7a}
  \mbox{non-empty, smooth, \red{connected} and compact for all $ \varepsilon\ge 0$ small enough.}
\end{equation}%
If, again
for each $\tau$ and
for \red{all} $\varepsilon$ small enough,
\begin{equation}\label{30III19.7}
 \mbox{the null geodesics normal to $\red{S_{\tau,\varepsilon}}$ with $\dot \Omega|_{\red{S_{\tau,\varepsilon}}} <0$  never leave $\mcU$,}
\end{equation}
%
where $\dot \Omega$ denotes the derivative of $\Omega$ along the geodesic,
then the domain of outer communications
$I^+(\scrim) \cap I^{-}(\scrip)$ is simply connected.
\end{Theorem}

\begin{Remark}
 \label{R16V19.1}

 The key property of the collection of smooth compact manifolds $\red{S_{\tau,\varepsilon}}$ is that they are future-inner trapped, cover a neighborhood of $\scrim$, and satisfy \eqref{30III19.7}. There are many ways of obtaining such a family.

First, one could, e.g., assume that there exists a foliation of $\scrim$ by smooth compact submanifolds,
 each intersecting the generators of $\scrim$ exactly once, and extend this foliation to a foliation of a  neighborhood of $\scrim$ in $\bmcM$ by moving along a field of spacelike directions, leading to the  family  $\hyp_\tau$    of the theorem.

  Next, the manifolds $\red{S_{\tau,\varepsilon}}$ can then, e.g., be obtained by moving a $\tg$-distance $\varepsilon$ away from $\scrim$ along $\hyp_\tau$, rather than intersecting with the level sets of $\Omega$.

Finally, one could define an alternative $\tau$-foliation, as needed in the statement of the theorem,
 by extending the just-mentioned foliation of $\scrim$ into $\mcM$ along null rather than spacelike
  hypersurfaces $\hyp_\tau$, with the manifolds $\red{S_{\tau,\varepsilon}}$ obtained by moving cuts of $\scrim$ an affine-parameter-distance $\varepsilon$  along the generators of $\hyp_\tau$.

   In each case the resulting $\red{S_{\tau,\varepsilon}}$'s will be future-inner trapped for $\varepsilon$ small enough.
 But note that in each case the additional condition \eqref{30III19.7} needs to be imposed.
\qed
\end{Remark}

\begin{figure}
\begin{center}
\mbox{
\includegraphics[width=1.3in]{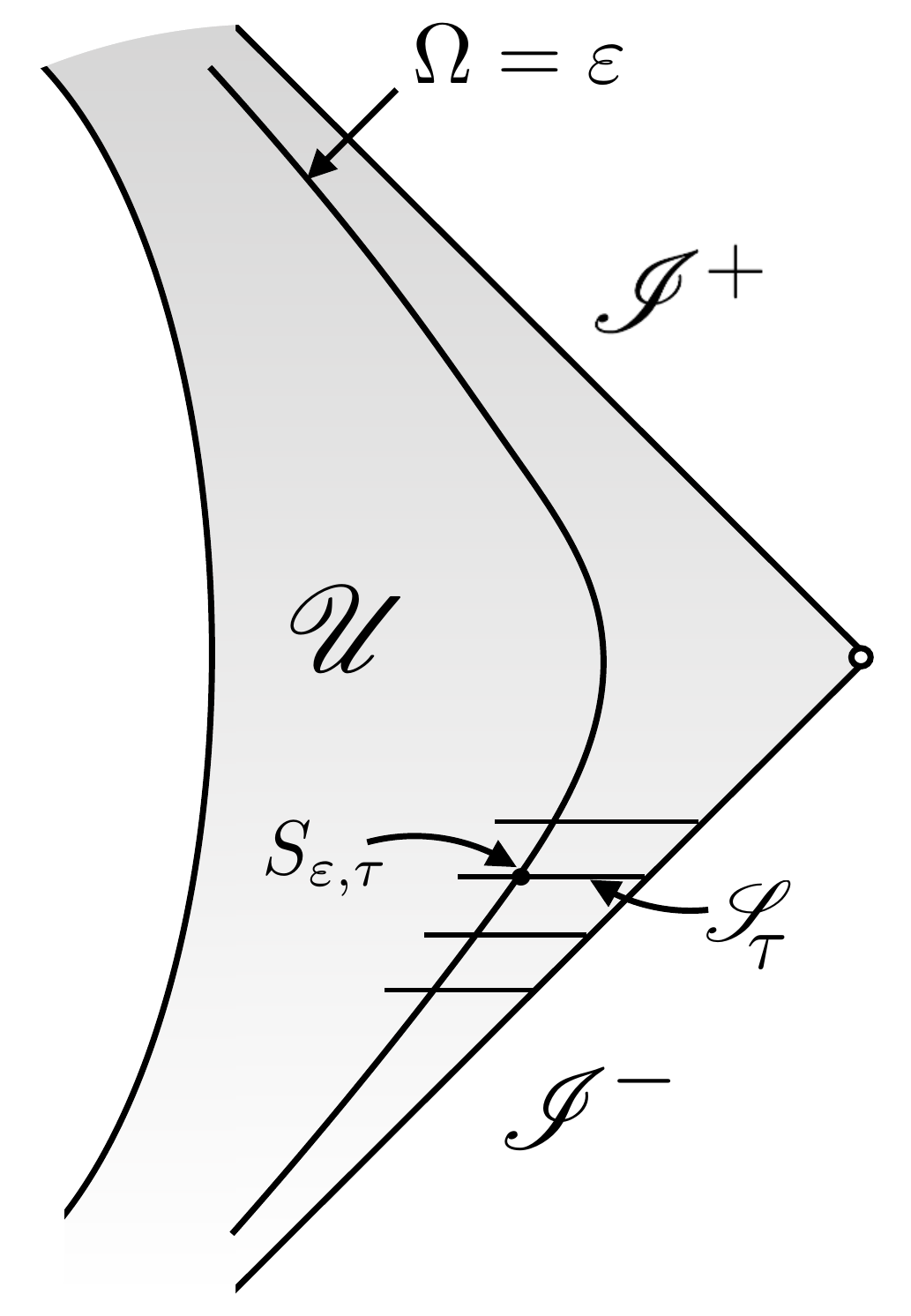}
}
\end{center}
\caption{ The sets $\mcU$, $\hyp_\tau$ and $\red{S_{\tau,\varepsilon}}$.
\label{F30III19.1}}
\end{figure}

\red{
\begin{Remark}
 \label{R12IV19.1}
The proof of  Theorem~\ref{T22III19.1} actually requires the weaker, but somewhat less  transparent condition that for each $\tau$ and for all $\epsilon$ small enough
\begin{equation}\label{30III19.7+2-}
 \mbox{
those generators of $\dot J^+( \red{S_{\tau,\varepsilon}})$ on which  $\frac{d\Omega}{ds}|_{\red{S_{\tau,\varepsilon}}} <0$  never leave $\mcU$.
}
\end{equation}
\Eqref{30III19.7} clearly implies \eqref{30III19.7+2-}.
\qed
\end{Remark}
}

\begin{Remark}
\label{R30III19.1}
For simplicity we assume that both metrics $\tilde\fourgKK$ and $\fourgKK$ are smooth up to boundary. One can check, using e.g.~\cite{ChGrant} and the conformal transformation formula for the divergence of null hypersurfaces,  that the proof applies if $\fourgKK$ is $C^2$ on $\mcM$ and $\tilde\fourgKK$ is $C^1$ up-to-boundary.
\qed
\end{Remark}

\begin{Remark}
 \label{R3IV19.1}
The theorem remains true, with an essentially identical proof, when the null energy condition is replaced by the following integrated null energy condition: for all half-geodesics $\gamma$ that start near $\scrim$ it holds that
$$
  \int_0^{\infty} \mathrm{Ric}(\gamma’,\gamma’) ds \ge 0
 \,;
$$
compare~\cite{GSWW}.
\qed
\end{Remark}

\noindent{\sc Proof of Theorem~\ref{T22III19.1}.}
The proof uses a mixture of well-known-by-now ideas from \cite{FriedmanSchleichWitt,galloway-topology,CGS}, with the heart of the argument stemming from~\cite{FriedmanSchleichWitt}. For the sake of clarity we present most details,
without claim to novelty.

Suppose that the domain of outer communications $\doc$ is non-empty and not simply-connected, otherwise there is nothing to prove.

Let $\bmcM_1$ be the universal cover of $\doc \cup \scri$
equipped  with the covering metrics  which we continue to denote by $\fourgKK$  and $\tilde \fourgKK$. The covering manifold $\bmcM_1$ contains several copies,
say $\mcU_\alpha$, with $\alpha \in \mathrm A$ for some index set $ \mathrm A$,
 of the simply connected neighborhood  $\mcU$ of the original $\scri$.  Each $\mcU_\alpha$ comes with its adjacent conformal boundaries $\scrip_\alpha$, $\scrim_\alpha$ and $\scri_\alpha:=\scrip_\alpha\cup\scrim_\alpha$.

Let us choose two distinct lifts of $\mcU$, call them $\mcU_1$ and $\mcU_2$, with conformal boundaries $\scri_1$ and $\scri_2$, and suppose that there exists a causal curve,  say $\tilde \gamma$, connecting $\scrim_1$ with $\scrip_2$.
Let $\tilde \tau$ be such that $\tilde \gamma$ meets $\scrim_1$ at $\hyp_{\tilde \tau }$. We can choose a small parameter $\epsilon>0$ and a value of $\tau_0$ near $\tilde \tau $ so that

\begin{enumerate}
  \item $\tilde \gamma$ intersects $\hyp_{\tau_0}\cap \{\Omega \le \epsilon\}\cap \mcM$, and
  \item $S:= \hyp_{\tau_0} \cap \{\Omega=\epsilon\}$  is compact, and
  \item   \eqref{30III19.7} holds, and
  \item $S$  is  inner future trapped, in the sense that it has negative divergence on those null geodesics which initially move normally away from $S$ to the future and in the direction of increasing $\Omega$.
        This follows from the following transformation formula%
\footnote{We take this opportunity to correct Equation (4.10) in \cite{CGS}, which should be replaced by \eqref{25V19.1}. This change has no incidence on the remaining arguments there.
}
 for the future- and past-divergences $\red{\theta^\pm}$, in space-time dimension $n+1$,
\begin{equation}\label{25V19.1}
\red{\theta^\pm} = (n-1) ( \tilde  T   \pm \tilde N)( \Omega) + \Omega\, \tilde \red{\theta^\pm}
 \,,
\end{equation}
where $\tilde \red{\theta^\pm} $ are the divergences calculated using the metric $\tilde g$, $\tilde T$ is the field of unit future-directed normals to $\hyp_{\tau_0}$ and $\tilde N$ is the field of outer-pointing normals to the level sets of $\Omega$ within $\hyp_{\tau_0}$.
By hypothesis $\nabla \Omega$ is null and past-directed on $\scrim$, tangent to $
\scrim$, while  $( \tilde  T   \pm \tilde N)$ is null, future-directed and transverse to $\scrim$. This shows that $\red{\theta^\pm}$ is positive, bounded away from zero on $\hyp_{\tau_0}\cap \scrim$, hence positive for all $\epsilon$ small enough.
\end{enumerate}

Let $S_1$ denote the lift of $S$ to $\mcU_1$.  We claim that $S_1$ is a future outer trapped surface as seen from $\scri_2^+$.
For this, note that the  {achronal} boundary $\dot J^+(S_1)$   forms near $S_1$ a union of two  {smooth} null hypersurfaces {ruled} by null geodesic  segments normal to $S_1$. Let us denote by
$\mcN_1^L$ the subset of $\dot J^+(S_1)$ on which $\Omega$ initially increases when moving away from $S_1$ along the generators, and $\mcN_1^R$ the subset on which $\Omega$ initially decreases.
As in~\cite[Theorem~6.1]{CGS} there exists a null geodesic $\gamma$ in $\dot J^+(S_1)$ connecting $S_1$ to $\scrip_2$. The geodesic $\gamma$ intersects $\partial \mcU_1$, and therefore cannot lie on $\mcN_1^R$   which is entirely included in $\mcU_1$ by \eqref{30III19.7}.

Hence $\gamma$ lies on $\mcN_1^L$,
which has negative divergence on $S_1$, as claimed.
But this contradicts Theorem~\ref{T12IV19.1}.
We conclude that
\begin{equation}\label{3IV19.61}
 I^{-}(\scrip_\alpha)\cap I^+(\scrim_\beta) = \emptyset
 \quad
 \mbox{for $\alpha \ne \beta$.}
\end{equation}

To finish the proof, note that the collection of open sets
$$
 I^{-}(\scrip_\alpha)\cap I^+(\scrim_\alpha)
 \quad
  \alpha \in \Omega
$$
covers $\bmcM_1$. The sets are pairwise disjoint by \eqref{3IV19.61}. Since $\bmcM_1$ is connected there is only one such set, contradicting the assumption of non-simple connectivity, and Theorem~\ref{T22III19.1} is established.
\qedskip
%

The least satisfactory condition in Theorem~\ref{T22III19.1} appears to be \eqref{30III19.7}, and the question arises whether it can be replaced by other natural conditions. We  list here some alternatives; we have not been able to verify that  \eqref{30III19.7} can be gotten rid of altogether.

\begin{Proposition}
 \label{P12IV19.1}
 Under the remaining hypotheses of Theorem~\ref{T22III19.1}, suppose that instead of \eqref{30III19.7} we have
 \begin{enumerate}
   \item    for each   $\tau $ and for \red{some} $\varepsilon$ small enough
   the connected component of
   the level set $\{\Omega=\varepsilon\}$ containing $\red{S_{\tau,\varepsilon}}$  (cf.\ Figure~\ref{F30III19.1})
    is a timelike hypersurface entirely contained in $\mcU$ and separating $\mcU$
 (compare Lemma~1 in~\cite{FriedmanSchleichWitt});~or
\newpage
   \item
    the set
\begin{equation}\label{15IV11.1+}
 \bmcD{}^- := \mcD^- \cup \scri^-
 \,,
  \
  \mbox{ where $\mcD^- = I^+(\scri^-)\cap \mcM$}
\end{equation}
is also globally hyperbolic, and
there exists a foliation $
\hyp_\tau$ as in Theorem~\ref{T22III19.1} such that for all $\tau$ and for \red{some} $\varepsilon$ small enough we have
\begin{equation}\label{30III19.7+1}
  \mbox{$\overline{\mcD^+
   \Big((\hyp_\tau\cap \{\Omega\le \varepsilon\})\cup \scrim_{\hyp_\tau} \Big)} \subset \mcU$,}
\end{equation}
where  $\scrim_{\hyp_\tau}:= \scrim\cap J^+(\hyp_\tau )$ (cf.\ Figure~\ref{F30III19.1+}).
\begin{figure}
\begin{center}
\mbox{
\includegraphics[width=1.8 in]{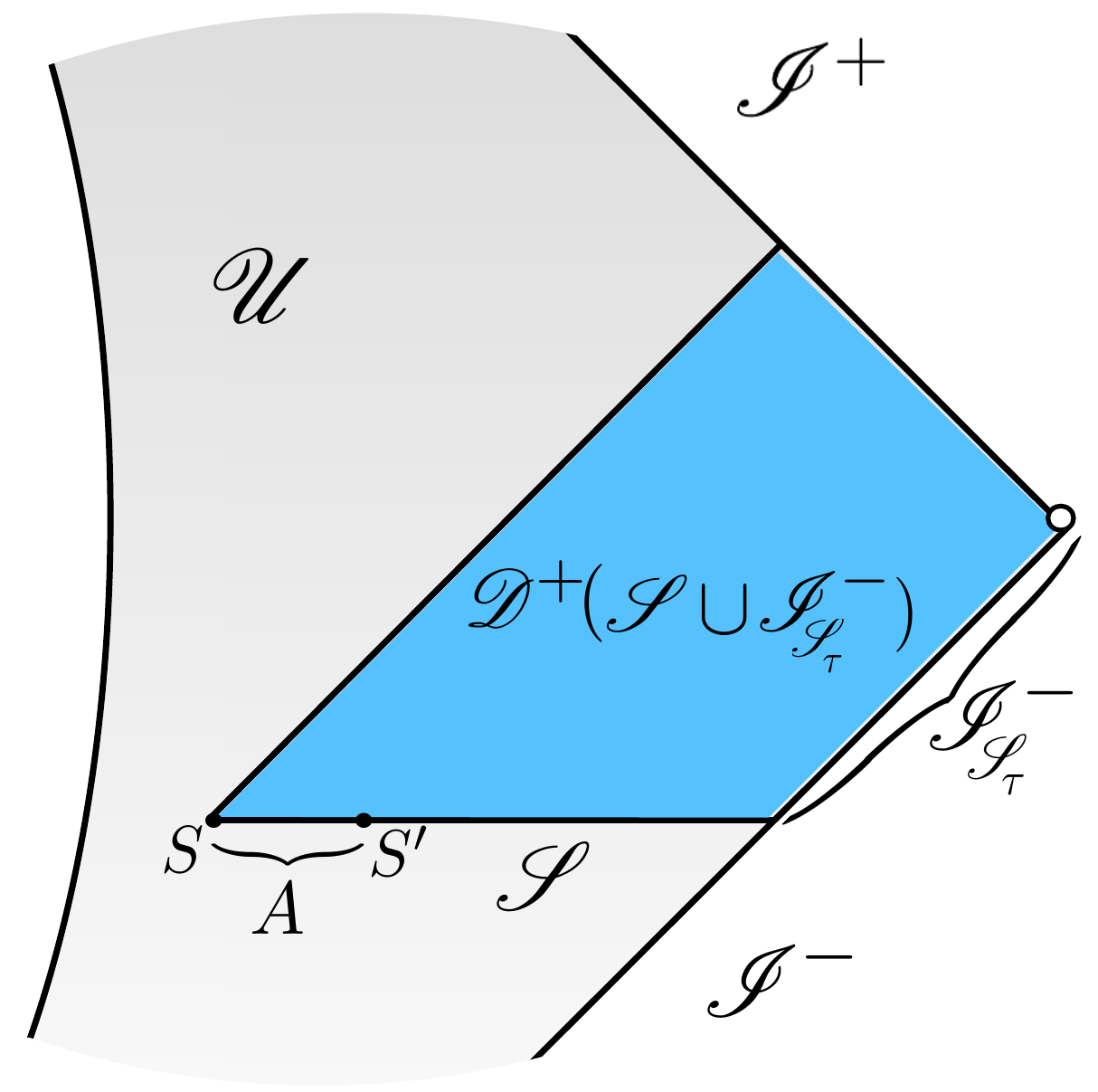}
}
\end{center}
\caption{The sets $\mcU$, $S \equiv \red{S_{\tau,\varepsilon}}$,  $\hyp \equiv \hyp_{\tau,\varepsilon}:=\hyp_{\tau} \cap \{\Omega \le \epsilon\}$, $\scrim_{\hyp_{\tau}}$ and $\mcD^+(\hyp\cup \scrim_{\hyp_\tau})$.\label{F30III19.1+}}
\end{figure}
 \end{enumerate}
Then the conclusions of  Theorem~\ref{T22III19.1}  hold.

\end{Proposition}

\proof
1.
We wish to show that \eqref{30III19.7+2-} holds. For this, let us denote by
$$
 \{\Omega=\varepsilon\}_0
$$
that  component of $\{\Omega=\varepsilon\}$    which contains $\red{S_{\tau,\varepsilon}}$. It is, by hypothesis, a  timelike hypersurface.
It contains a compact achronal spacelike surface $\red{S_{\tau,\varepsilon}}$ and therefore constitutes a globally hyperbolic  space-time of its own, with Cauchy surface $\red{S_{\tau,\varepsilon}}$, by \cite{BILY,Galloway:cauchy}. Since $\{\Omega=\varepsilon\}_0$ separates $\mcU$ by hypothesis,  every generator of $\dot J^+(\red{S_{\tau,\varepsilon}})$ which starts with $\dot \Omega <0$ at $\red{S_{\tau,\varepsilon}}$ and which leaves $\mcU$ has to cross $\{\Omega=\varepsilon\}_0$ again. But global hyperbolicity of $\{\Omega=\varepsilon\}_0$ implies that through every point of  $\{\Omega=\varepsilon\}_0$ there exists a timelike curve meeting $\red{S_{\tau,\varepsilon}}$, which is not possible since the generators of $\dot J^+(\red{S_{\tau,\varepsilon}})$
do not intersect $I^+(\red{S_{\tau,\varepsilon}})$.
\medskip

2.
Let $\tau_0$, as well as the remaining
notation, be as in the proof of Theorem \ref{T22III19.1}.

Let
$\hyp_{\tau_0,\epsilon } = \hyp_{\tau_0} \cap \{\Omega \le \epsilon \}$,
and set  $\widehat{\hyp}  = \hyp_{\tau_0,\epsilon } \cup \scri^-_{\hyp_{\tau_0}}$.
From (3.5) we have,
\beq
\overline{\mcD^+(\widehat{\hyp})} \subset \mcU.
\eeq
Note that the manifold $S\equiv \hyp_{\tau_0} \cap \{\Omega=\epsilon\}$  of the proof of Theorem~\ref{T22III19.1} satisfies $S = \mbox{\rm edge}(\widehat{\hyp})$.  Let $A \subset \hyp_{\tau_0,\epsilon } \cap \mcM$ be a smooth compact manifold diffeomorphic to $[0,1] \times S$ with boundary $\partial A = S \cup S'$, see Figure~\ref{F30III19.1+}.

Suppose $\langle\langle\scri\rangle\rangle$ is not simply connected.  Let $\mcU_1$ and $\mcU_2$ be as in the proof of Theorem~\ref{T22III19.1}.  Let $S_1$, $S_1'$, \blue{$\hyp_1$},
 $\widehat{\hyp}_1$
 and $A_1$ be the lifts of $S$, $S'$, \red{$\hyp_{\tau_0,\epsilon } $}, $\widehat{\hyp}$ and $A$, respectively, to
 $\mcU_1$.  We may assume $J^+(A_1)$ meets $\scri_2^+$
but, by Lemma 2.4, does not contain all of it. Then there exists a future complete null geodesic $\g$ on $\dot{J}(A_1)$ that extends from $\scri_2^+$ to  either $S_1$ or $S_1'$ when followed to the past.

Suppose, first, that $\gamma$ extends to $S_1'$. Then it will have to cross
$D^+(\widehat{\hyp}_1) \subset \mcU_1$  at some point, say
$p \in H^+(\widehat{\hyp}_1) \cap \mcM$ (not on $S_1$).  Let us denote by $\hat \gamma$ a generator of
$H^+(\widehat{\hyp}_1)$ passing through $p$.

Now, generators of $ H^+(\widehat{\hyp}_1)$ are either past inextendible, or have an end point on the edge of  $\widehat{\hyp}_1$. We wish to show that $\hat \gamma$  extends to $S_1$, and therefore
the former case cannot occur. In order to see this,
we first note  that any point $q$ lying on a generator of $ H^+(\widehat{\hyp}_1)$ is in $ J^+(\hyp_1
,\bmcD{}^- _1)$, where
$$
\bmcD{}^- _1 := \mcD^-_1 \cup \blue{\scri^-_1}
 \,,
  \
  \mbox{ with $\mcD^-_1 := I^+(\blue{\scri^-_1})\cap \bmcM$}
  \,.
$$
Indeed, since $q \in H^+(\widehat{\hyp}_1)$ there exists a causal curve connecting $q$ with $\widehat \hyp_1= \hyp_1
\cup \scrim_{\hyp
}$. If the curve connects $q$ with $\hyp_1
$ there is nothing to prove; otherwise it connects to $\scrim_{\hyp_1
}$, but then we can concatenate it with a segment of generator of $\scrim_1$ which slides down from the intersection point to $\scrim_{\hyp_1
}$.
We conclude that the part of the generator of  $ H^+(\widehat{\hyp}_1)$ which lies to the past of $p$ is included in
%
$$
 \mcK:=  J^-(p,\bmcD{}^-_1)\cap
  J^+ ( \hyp_1
  ,\bmcD{}^-_1)
  \,.
$$
The set $\mcK$ is compact  by global hyperbolicity of   $\bmcD{}^-_1$.
Since $\hat \gamma$ is contained in the compact set $\mcK$, it has to end on the edge of $\widehat{\hyp}_1$ to avoid a violation of strong causality, as claimed.

Now, since the edge of $\widehat \hyp_1$ is $S_1$,  $\hat \gamma$ clearly enters the timelike future of $A_1$.
But this implies that
$\g$ enters the timelike future of $A_1$, which is not possible for a curve lying on $\dot J^+(A_1)$, a contradiction.  Thus $\g$ extends to $S_1$, and necessarily to the inside.  Hence  $\g$ lies on $\mcN_1^L$, and the proof continues as in the proof of Theorem \ref{T22III19.1}.

\qed

\bigskip

{\noindent \sc Acknowledgements} The research of PTC was supported in part by a grant FWF P 29517-N27, and by a grant of the Polish National Center of Science (NCN) 2016/21/B/ST1/00940. The research of GG was supported in part by an NSF grant DMS-1710808.

{\small
\bibliographystyle{amsplain}

\bibliography{ChruscielGallowayStandardTopologicalCensorship-minimal}
}
\end{document}